\documentclass[12pt]{article}
\usepackage{amsmath}
\usepackage{graphicx,psfrag,epsf}
\usepackage{enumerate}

\usepackage[round]{natbib}
\usepackage{url} % not crucial - just used below for the URL 
\usepackage{amsfonts}
\usepackage{amsthm}
\usepackage{mathrsfs}
\usepackage{multirow}
\usepackage{color}
\usepackage[titletoc,title]{appendix}
\usepackage[font={footnotesize}]{caption}
\usepackage[font={footnotesize}]{subcaption}
\usepackage{tabularx} 

\RequirePackage{hyperref}
\usepackage{longtable}

\newtheorem{theorem}{Theorem}[section]
\makeatletter
\renewcommand{\author}[1]{\gdef\@author{\parbox{\textwidth}{\raggedright #1}}}
\makeatother

\newcommand{\bm}{\mathbf{m}}

\newcommand{\bX}{\mathbf{X}}

\begin{document}
	
\bibliographystyle{unsrtnat}
	
\def\spacingset#1{\renewcommand{\baselinestretch}{#1}\small\normalsize} \spacingset{1}
	
\title{A Bivariate Poisson-Gamma Distribution: Statistical Properties and Practical Applications}
\author{
{\small Indranil Ghosh$^{1}$, Mina Norouzirad$^{2}$\thanks{Corresponding author: m.norouzirad@fct.unl.pt}, Filipe J. Marques$^{2,3}$}\\[2em]
{\footnotesize $^{1}$ Department of Mathematics and Statistics, University of North Carolina, Wilmington, USA}\\
{\footnotesize $^{2}$Center for Mathematics and Applications (NOVA Math), NOVA School of Science and }\\ 
{\footnotesize Technology (NOVA FCT), NOVA University Lisbon, Caparica, Portugal}\\
{\footnotesize $^{3}$Department of Mathematics, NOVA School of Science and Technology (NOVA FCT),}\\ 
{\footnotesize NOVA University Lisbon, Caparica, Portugal}
}
\date{}
\maketitle
	
\begin{abstract} 
Although the specification of bivariate probability models using a collection of assumed conditional distributions is not a novel concept, it has received considerable attention in the last decade.  In this study, a bivariate distribution—the bivariate Poisson-Gamma conditional distribution—is introduced, combining both univariate continuous and discrete distributions. This work explores aspects of this model's structure and statistical inference that have not been studied before. This paper contributes to the field of statistical modeling and distribution theory through the use of maximum likelihood estimation, along with simulations and analyses of real data.\\

\end{abstract}
	
\section{Introduction}\label{Introduction}
Many different areas of study have made use of bivariate distributions to delve into real-world phenomena; they include environmental data, cluster analysis, actuarial science, under- and over-dispersed data modeling, Bayesian statistics, and many more. Important things that should be mentioned are listed below: Three bivariate Poisson binomial distribution types by \citet{Charalambides:Papageorgiou:1981}; two variants of a bivariate Hermite distribution, by \citet{Kemp:Papageorgiou:1982}, one with eight parameters and the other with five. \cite{Kocherlakota:Kocherlakota:1992} conducted a comprehensive survey on various types of bivariate and multivariate discrete distributions, highlighting their applications across fields including plant pathology, physical sciences, and accident statistics. Indeed, the list is incomplete. The univariate Poisson distribution is widely used in modeling count data, particularly for rare events. Therefore, there is a significant amount of statistical literature addressing the application of the Poisson distribution and its various mixtures, leading to bivariate and multivariate extensions. \citet{Greenwood:Yule:1920} are recognized as pioneers in introducing a blend of Poisson distributions that incorporate a gamma distribution for the Poisson parameter $\lambda$.

Continuing on, we provide a few helpful citations related to developments on higher-dimensional mixtures of discrete and continuous probability models. \citet{Bibby:Vaet:2011} developed a two-dimensional beta-binomial distribution employing a joint beta distribution proposed by \citet{Jones:2002} to represent the joint probability of successes. \citet{Sarabia:Gomez-Deniz:2011} suggested and examined multivariate versions of the beta mixture of Poisson distribution, noting that this class of distributions has a flexible covariance structure. \citet{Karlis:Xekalaki:2005} explored novel characteristics of Poisson mixtures. 

Next, we concentrate on a specific type of mixture distributions that includes Poisson and gamma distributions. Following are a number of relevant sources: In studying the correlation structure between accidents and accident proneness, \citet{Arbous:Kerrich:1951} examined the role of a mixture of Poisson and gamma. Based on a combination of Poisson and gamma distributions, \citet{Bates:Neyman:1952} examined the relationship between light and severe accidents. \cite{Nelson:1985} built upon the research of \cite{Arbous:Kerrich:1951} by including a Dirichlet distribution into the analysis of cross-sectional data and permitting individual rates to change at the start of the second period in the study of two-period longitudinal data. \cite{Lemaire:1979} examined the suitability of combining Poisson and gamma distributions in developing a bonus-malus system with an exponential utility function. In a related study, \cite{Piegorsch:1990} focused on maximum likelihood estimators for the dispersion parameters within a combination of Poisson and gamma distributions. Moreover, \cite{Gomez-Deniz:Calderin-Ojeda:2014} proposed a new bivariate Conway-Maxwell Poisson distribution for application in risk theory, using conditional specifications in its derivation.

Remarkably, no bivariate combination of Poisson-gamma distribution was taken into consideration from the perspective of conditional specification in any of the aforementioned references. This set-up has two conditionals: One is a Poisson whose intensity parameter adjusts depending on the conditioning variable. The other is a two-parameter gamma whose shape and scale parameters also change based on the conditioning variable.  In this study, we investigate a bivariate Poisson-gamma distribution derived from two conditional distributions initially proposed by \citet{Arnold:1999}. There may be a good reason to study this bivariate distribution if this probability model has rich enough correlation and marginal over-dispersion. Subsequently, we present fundamental preliminaries about the previously described bivariate Poisson-gamma conditionals (BPGC) distribution along with real-world application that has not been discussed earlier in the literature.

The remainder of the paper is structured as follows:  Section \ref{sec:2} presents the BPGC distribution along with its structural characteristics. Within a classical framework, the estimation of model parameters is discussed in Section \ref{sec:Mle}. For the sake of illustration, Section \ref{sec:4} shows a thorough simulation study for a number of common model parameter choices. Section \ref{sec:5} examines a real-world data set to show the practical application of the BPGC distribution. Finally, Section \ref{sec:6} offers some concluding remarks.

\section{Structural Properties of the Bivariate Poisson-Gamma Conditional Distributions}\label{sec:2}

Consider a bivariate random vector $\left(X,Y\right)$, where both variables are defined on $\mathbb{R^{+}}$. According to this,
\begin{itemize}
\item The conditional distribution $X|Y = y$ adheres to a Poisson distribution for each fixed $Y=y$, and its mean parameter defined is $\exp\{m_{10} - m_{11}y + m_{12} \log(y)\}$. Specifically,
\begin{equation}\label{eq:dist_X|Y}
    X|Y=y\sim {\rm Poisson}\left(\exp\{m_{10} - m_{11} y + m_{12} \log(y)\}\right),
\end{equation}
where $m_{10}$ is positive and both $m_{11}$ and $m_{12}$ are non-negative.

\item  An example of a Gamma distribution with shape parameter $m_{02} + m_{12}x$ and scale parameter $m_{01} + m_{11} xy$ would be the conditional distribution $Y|X = x$ for every fixed $X=x$. %For each fixed $X=x$, the conditional distribution $Y|X = x$ follows a Gamma distribution with a shape parameter of $m_{02} + m_{12}x$ and a scale parameter of $m_{01} + m_{11} xy$. 
That is,
\begin{equation}\label{eq:dist_Y|X}
    Y|X=x\sim {\rm Gamma}\left(m_{02}+m_{12}x, m_{01}+m_{11}xy\right).
\end{equation}
Both $m_{01}$ and $m_{02}$ are positive, while both $m_{11}$ and $m_{12}$ are non-negative.
\end{itemize}

Subsequently, in accordance with \citet{Arnold:1999}, the general formulation of the joint distribution of the BPGC, which is part of a broad category of conditionals within the Exponential families, is expressed as\begin{eqnarray}\label{eq:BPGC}
     f\left(x,y|\bm\right) &=& \left(x! y\right)^{-1}\exp\bigg(c + m_{10}x-m_{01}y-m_{11}xy  +m_{02}\log y+m_{12}x\log y\bigg),\cr 
    && \hspace{7cm} x=0,1,\ldots, \quad  y>0, \quad  
\end{eqnarray}
with $\bm = \left(m_{10}, m_{01}, m_{11}, m_{02}, m_{12}\right)$. Let $m_{01}, m_{02}, m_{10} > 0$, and $m_{11}, m_{12} \geq 0$. The normalizing constant $c$ is defined as
$$
c = - \log  \sum_{x = 0}^{\infty} \int_{0}^{\infty} (x! y)^{-1}\exp\bigg(m_{10}x-m_{01}y-m_{11}xy  +m_{02}\log y+m_{12}x\log y\bigg) {\rm d}y .
$$

Several practical structural characteristics and their applications will be discussed in this paper for the bivariate distribution with joint density in \eqref{eq:BPGC}, which comprises both continuous and discrete random variables, enabling the modeling of several phenomena. Such forms of bivariate distributions are rare in the statistical literature. For example, identifying bivariate distributions in which one marginal is continuous while the other is discrete is challenging. As \citet{Arnold:1999} points out, the BPGC distribution has two interesting special cases:
\begin{itemize} 
\item The distribution aligns with the compound Poisson distribution when \( m_{11} = 0 \) and \( m_{12} = 1 \) \citep{Feller:1957}. 
\item If $m_{11} = m_{12} = 0$, then $X$ and $Y$ are independent.
\end{itemize}

The R library \texttt{BPGC} is made available at \url{https://github.com/mnrzrad/BPGC} in an effort to improve the accessibility of the BPGC distribution. The function \texttt{dBPGC} in this library computes the BPGC distribution at an arbitrary point $(x, y)$. The next part of this section will look at some useful structural characteristics of the joint distribution shown in \eqref{eq:BPGC}.

The marginal probability mass function (p.m.f.) of the discrete random variable $X$ is defined as
\begin{eqnarray}\label{eq:h(X)}
    f_{X}(x | \bm) &=& \frac{ \exp\left(c + m_{10} x \right)}{x!} \int_{0}^{\infty} y^{m_{02} + m_{12}x - 1}\exp\left(-y \left(m_{01} + m_{11}x \right)\right) dy\cr 
    && \cr &&\cr 
    &=&\frac{\Gamma\left(m_{02} + m_{12} x \right)\exp\left(c + m_{10}x \right)}{x! \left(m_{01} + m_{11} x \right)^{m_{02} + m_{12} x}}, \quad x = 0, 1, \ldots,
\end{eqnarray}
where $\Gamma(m)=\int_{0}^{\infty}x^{m-1} \exp(-x)dx$ denotes the standard gamma function. Note that the generalization of the negative binomial distribution forms the marginal p.m.f. for $X$ in \eqref{eq:h(X)}.

Likewise, the marginal density of the continuous random variable $Y$ can be found as follows
\begin{eqnarray}\label{eq:g(y)}
    f_Y(y | \bm) & = & \frac{\exp\left(c - m_{01} y + m_{02} \log(y)\right)}{y}
    \sum_{x = 0}^{\infty}\frac{\exp\left(x \left(m_{10} - m_{11}y + m_{12} \log(y)\right)\right)}{x!}\cr 
    && \cr && \cr 
    & = &\frac{\exp\left(c - m_{01} y + m_{02} \log(y)\right)}{y} \sum_{x = 0}^{\infty}\frac{\bigg(\exp\left(m_{10} - m_{11}y \right) y^{m_{12}}\bigg)^{x}}{x!}\cr 
    &&\cr &&\cr 
    & = & y^{m_{02} - 1}\exp\left(\exp\left(m_{10} - m_{11}y \right) y^{m_{12}} + c - m_{01} y \right), \quad y > 0. 
\end{eqnarray}
In Figure \ref{fig:0}, plots of joint density in \eqref{eq:BPGC} for different parameter choices are shown.

\begin{figure}[!htbp]
    \centering
    \begin{tabular}{cc}
        \includegraphics[scale = 0.32, trim={90mm 0mm 90mm 0mm},clip]{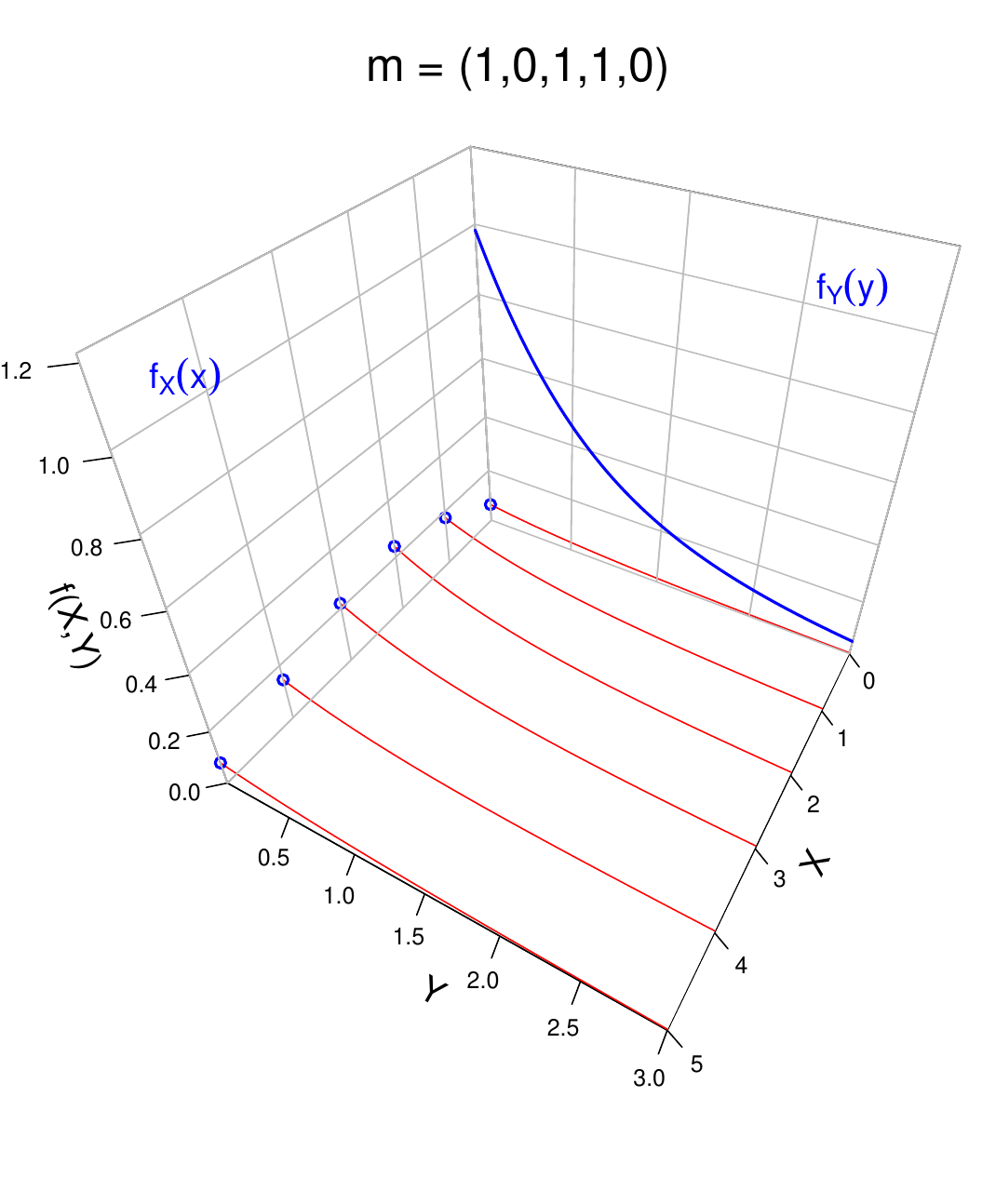} &       
        \includegraphics[scale = 0.32, trim={90mm 0mm 90mm 0mm},clip]{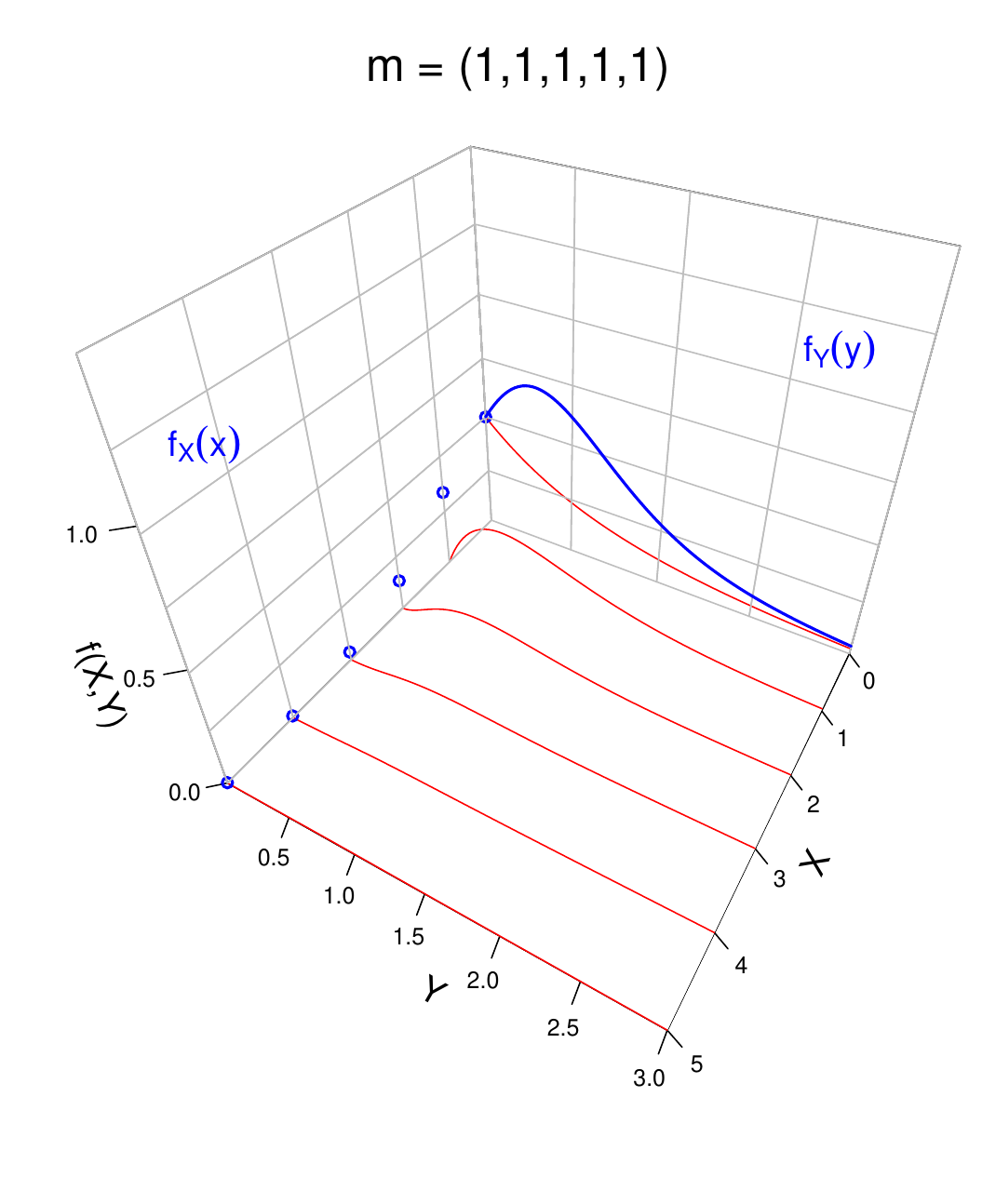} 
        \\
        \includegraphics[scale = 0.32, trim={90mm 0mm 90mm 0mm},clip]{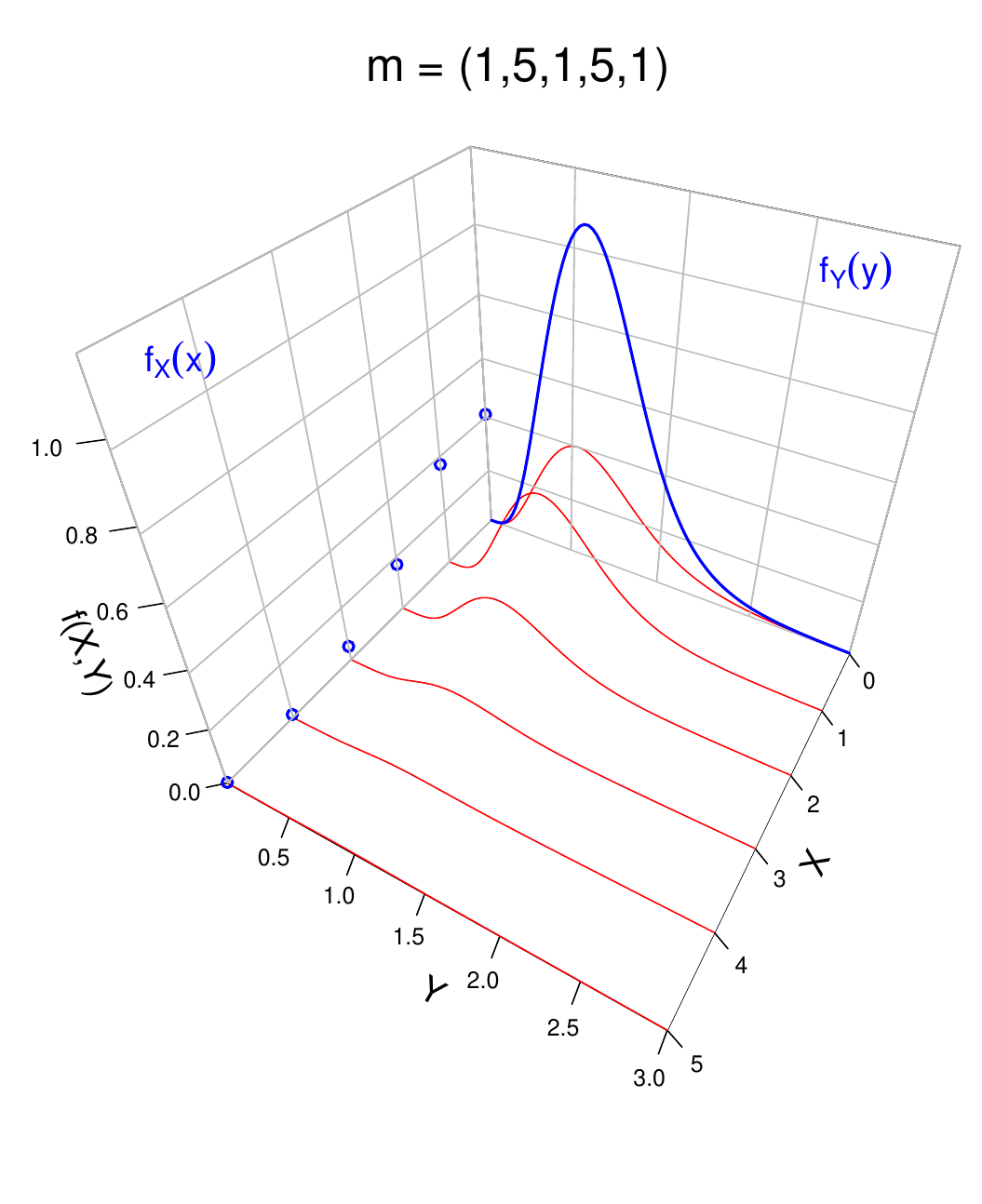} & \includegraphics[scale=0.32, trim={90mm 0mm 90mm 0mm},clip]{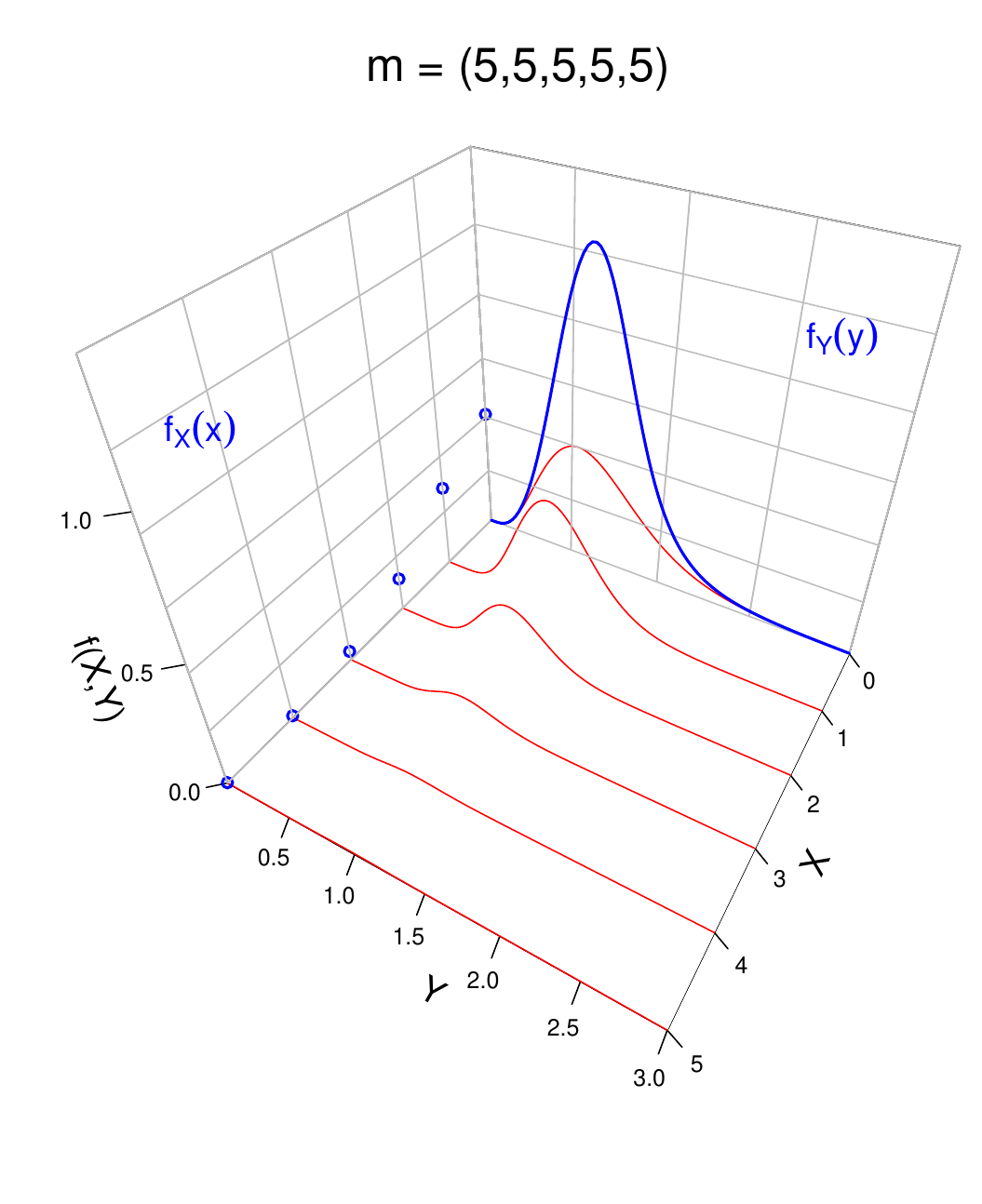}\\
    \end{tabular}
     \caption{Representative plots of  the joint and marginal distributions of BPGC for various choices of the model parameters.}\label{fig:0}
\end{figure}

\subsection{Total positivity of order two property}
Consider real numbers $x_{1}$, $x_{2}$, $y_{1}$, and $y_{2}$ where $0 < x_{1} < x_{2}$ and $0 < y_{1}< y_{2}$. Hence, the pair $\left(X,Y\right)$ is characterized by the total positivity of order two (TP$_{2}$) property if and only if
\begin{equation}\label{eq:TP2}
\begin{vmatrix}
f(x_{1}, y_{1} | \bm) & f(x_{1}, y_{2}| \bm)\\
f(x_{2}, y_{1} | \bm) & f(x_{2}, y_{2}| \bm)
\end{vmatrix}  \geq 0.
\end{equation}
The condition in \eqref{eq:TP2} for the joint distribution specified in \eqref{eq:BPGC} simplifies to
\begin{equation*}
    \exp\left( m_{12} \left( x_2 - x_1 \right) \left( \log(y_2) - \log(y_1) \right) \right)\, \exp\left(- m_{11} \left(x_2 - x_1 \right) \left( y_2 - y_1 \right)\right) \geq 0.
\end{equation*}
This holds true provided that $ 0 < x_{1} < x_{2}$ and $0 < y_{1} < y_{2}$. As a consequence, the joint distribution in \eqref{eq:BPGC} has the ${\rm TP}_{2}$ property which leads to the following implications:
\begin{itemize}
\item For any $x$, the conditional probability $P\left(X\leq x|Y=y\right)$ remains constant or decreases as $y$ increases.
\item For every $y$, the conditional probability $P\left(Y\leq y|X=x\right)$ remains constant or decreases as $x$ increases. 
\item For all $y$, the conditional probability $P\left(Y>y|X>x\right)$ remains constant or increases as $x$ increases.
\item The conditional probability $P\left(Y\leq y|X\leq x\right)$ is greater than or equal to $P\left(Y\leq y\right)P\left(X\leq x\right)$,
\item The conditional probability $P\left(Y>y|X>x\right)$ is greater than or equal to $P\left(Y>y\right)P\left(X>x\right)$.
\end{itemize}

\subsection{Exponential family}

\begin{theorem}
The distribution ${\rm BPGC}(\bm)$ is a member of the exponential family with five parameters.
\end{theorem}

\begin{proof}
The BPGC distribution in \eqref{eq:BPGC} is represented in the form
\begin{equation*}
    f\left(x,y|\bm \right) = \exp \left( \exp \Bigg( \sum_{j=1}^{5} {\delta_{j} \left( \bm \right) U_j(x,y)} - \Delta\left( \bm \right)\Bigg) \right),
\end{equation*}
where $\delta_{1}\left( \bm \right) = m_{10}$, $\delta_{2}\left(\bm \right) = -m_{01}$, $\delta_{3}\left( \bm \right) = -m_{11}$, $\delta_{4}\left(\bm \right) = m_{02}$, $\delta_{5}\left(\bm \right) = m_{12}$, $U_{1}(x,y) = x$, $U_{2}(x,y) = y$, $U_{3}(x,y) = xy$, $U_{4}(x,y) = \log(y)$, $U_{5}(x,y) = x\log(y)$, and $\Delta\left( \bm \right) = c + \log\left(x! y\right)^{-1}$. Hence, the joint distribution $\left(X,Y\right) \sim {\rm BPGC}(\bm)$ belongs to the five-parameter exponential family.
\end{proof}

Consequently, for a random sample of size $n$ drawn from the ${\rm BPGC}(\bm)$ distribution, the statistics $\sum_{i}x_{i}$, $\sum_{i}y_{i}$, $\sum_{i}x_{i}y_{i}$,$\sum_{i}\log y_{i}$, and $\sum_{i}x_{i}\log y_{i}$ form a complete and sufficient statistic for $\bm$.

Distributions in the exponential family have several useful characteristics. For instance, mean, variance, covariance, and moment-generating functions can be easily obtained by differentiating $\Delta\left(\bm\right)$. Moreover, using the Lehmann-Scheffe Theorem (see \cite{Lehmann:Casella:2006} for details), it may be possible to derive a uniformly minimum-variance unbiased estimator (UMVUE) for the parameters, provided that an unbiased function of $T$ can be found. Even in cases where this is not achievable, the minimum-variance unbiased estimator (MVUE) can still be obtained through bias correction. For the definition and properties of a UMVUE and MVUE, interested readers are advised to read \citet{Lehmann:Casella:2006} and the references therein.

\subsection{Shape of the distribution}
A function of two variables reaches a critical point when its first-order partial derivatives are equal to zero. Identifying these critical points in a bivariate probability distribution serves two primary purposes: 1) to understand the overall shape of the distribution for evaluating its flexibility in modeling observed phenomena with similar patterns, and 2) to emphasize the importance of peakedness property in bivariate distributions.  Peakedness is a descriptive measure that indicates the concentration or dispersion of the distribution, particularly reflecting the fatness of the tails of the distribution. The goal here is to determine the unique mode of the specified bivariate probability distribution. With this in mind, let us now analyze the shape of the BPGC distribution as presented in \eqref{eq:BPGC}.

% A critical point of a function with two variables  occurs where  the partial derivatives of the first order are zero. Identifying these critical points of a bivariate probability distribution  has two main reasons: 1) to determine the  overall shape of the distribution in order to examine its flexibility in fitting an observed phenomenon exhibiting the same pattern, and 2) in bivariate distributions, the property of peakedness is important. It is a descriptive  measure indicating the concentration or dispersion of the distribution. Specifically, it measures the fatness of the tails of the  distribution function. The  specific goal is to  locate the unique mode function for the given bivariate probability distribution.  Now, let us  analyze the shape of the  BPGC distribution given in \eqref{eq:BPGC}.

A first derivative of the function $\log(f(x,y))$ with regard to $x$ can be written as
\begin{equation*}
    \frac{\partial \log(f(x, y|\bm))}{\partial x} = -\frac{\Gamma(x + 1) \psi^{(0)}(x+1)}{x!}+m_{10}-m_{11} y+m_{12} \log(y),
\end{equation*}
where $\psi^{(0)}(\cdot)$ is the polygamma function defined by 
$$
\psi^{(0)}(x) = \psi(z) = {\Gamma'(x)}/{\Gamma(x)}.
$$
Furthermore, the partial derivative with respect to $y$ is given by
\begin{equation*}
    \frac{\partial \log(f(x, y | \bm))}{\partial y} = \frac{- m_{01} y + m_{02} - m_{11} x y + m_{12} x - 1}{y}.
\end{equation*}
So, there could be several important points for this model. Figure \ref{fig:0} confirms this statement.

\subsection{Stochastic ordering}
In statistical literature, distribution ordering is important, especially when it comes to lifetime distributions. \cite{Johnson:Kotz:Balakrishnan:1995} extensively discussed the ordering of various lifetime distributions. This paper, four different stochastic orders are considered: the usual order, the hazard rate order, the mean residual life  order, and the likelihood ratio order for two independent random variables within a constrained parameter space. It is known that a family that has the monotone likelihood ratio characteristic also has likelihood ratio ordering. This property signifies that a uniformly most powerful test exists for any one-sided hypothesis when the other parameters are fixed.

Assuming independent random variables $X$ and $Y$ with cumulative distribution functions (c.d.f.'s) $F_{X}$ and $F_{Y}$, respectively, we have
\begin{enumerate}
    \item[1)] $X\geq_{\rm st}Y$ indicates that $X$ is greater than $Y$ in the stochastic order if, for all $x$, $F_{X}(x|\bm)\leq F_{Y}(x|\bm)$.
    \item[2)] $X\geq_{\rm hr}Y$ denotes that $X$ is greater than $Y$ in the hazard rate order if, for every $x$, $h_{X}(x|\bm)\leq h_{Y}(x|\bm)$, where $h_{X}(x|\bm)$ and $h_{Y}(y|\bm)$ are the hazard rate functions, given by
    \begin{eqnarray*}
    h_{X}(x|\bm) = \frac{f_{X}(x|\bm)}{1 - F_X(x|\bm)} \quad \text{and} \quad h_{Y}(y|\bm) = \frac{f_{Y}(y|\bm)}{1 - F_Y(y|\bm)}
    \end{eqnarray*}
    with $f_X(x|\bm)$ and $f_Y(y|\bm)$ as defined in \eqref{eq:h(X)} and \eqref{eq:g(y)}, respectively.
    \item[3)] $X\geq_{\rm mrl}Y$ indicates that $X$ is greater than $Y$ in the mean residual life order if, for all $x$, $m_{X}(x|\bm)\leq m_{Y}(x|\bm)$ where $m_{X}(x|\bm)$ and $m_{Y}(y|\bm)$ are the mean residual life functions given by
    \begin{equation*}
        m_{X}(x|\bm) = {\rm E}(X - t | X > t) \, = \, \frac{1}{1 - F_X(x|\bm)} \int_t^{\infty} F_X(u|\bm) {\rm d}u
    \end{equation*}
    and 
    \begin{equation*}
        m_{Y}(y|\bm) = {\rm E}(Y - t | Y > t) \, = \, \frac{1}{1 - F_Y(y|\bm)} \int_t^{\infty} F_Y(u|\bm) {\rm d}u,
    \end{equation*}
    respectively.
    \item[4)] $X\geq_{\rm lr}Y$ denotes that $X$ is greater than $Y$ in the likelihood ratio order if, for all $x$, ${f_{X}(x|\bm)}\,/\,{f_{Y}(x|\bm)}$ decreases in $x$.
\end{enumerate}

According to \cite{Shaked:Shanthikumar:1994}, the following well-known results establish the stochastic ordering of distributions:
\begin{eqnarray*}
(X\geq_{\rm lr}Y)\; \Rightarrow &(X\geq_{\rm hr}Y)& \Rightarrow \; (X\geq_{\rm mrl} Y)\nonumber\\ &\Downarrow\cr 
&(X\geq_{\rm st}Y)&
\end{eqnarray*}

The BPGC distribution is ordered according to the strongest ordering, i.e. the  likelihood ratio ordering, as proved in the following theorem. This result highlights the flexibility of the BPGC distribution.

\begin{theorem}\label{thm:1}
 Let $X$ and $Y$ follow marginal distributions $f_X(x|\bm)$ and $f_Y(y|\bm)$ as given in \eqref{eq:h(X)} and \eqref{eq:g(y)} respectively. If $m_{10}, m_{11}\geq 0,$ $m_{12}\geq 1$ and $m_{02}> 0$; then, $X\geq_{\rm lr} Y$, $X\geq_{\rm hr} Y$, $X\geq_{\rm mrl} Y$, and $X\geq_{\rm st} Y$.
 \end{theorem}
\begin{proof}
The likelihood ratio of $X$ to $Y$ is given by ${f_{X}(x|\bm)}\,/\,{f_{Y}(y|\bm)}$, which equals to
\begin{equation*}
    \frac{\Gamma \left(m_{02}+m_{12} x\right) \left(c+m_{11} x\right)^{m_{02}+m_{12} x}\exp\left(m_{01} y-m_{02} \log y+\exp(-m_{11} y) y^{m_{12}}\right)}{(x-1)!}.
\end{equation*}
Calculating the derivative of the natural logarithm of the likelihood ratio with respect to $x$ results in
\begin{eqnarray*}
    \frac{\partial}{\partial x} \left\{\log \left(\frac{f_{X}(x | \bm)}{f_{Y}(x | \bm)}\right)\right\} &=& -\frac{\Gamma (x) \psi^{(0)}(x)}{(x-1)!} + \frac{m_{11} \left(m_{02} + m_{12} x \right)}{c + m_{11} x}\cr && \cr 
    && \qquad + m_{12} \log\left(c + m_{11} x \right) + m_{12} \psi^{(0)} \left( m_{02} + m_{12} x\right)\cr 
    && \cr 
    & = & - \psi^{(0)} (x) + \frac{m_{11} \left(m_{02} + m_{12} x\right)}{c+m_{11} x} + m_{12} \log\left( c + m_{11} x \right)\cr 
    &&\cr 
    && \qquad + m_{12} \psi^{(0)} \left( m_{02} + m_{12} x\right), 
\end{eqnarray*}
where $\psi^{(0)}(\cdot)$ denotes the polygamma function, given by
$\psi^{(0)}(x) = \psi(z) = {\Gamma'(x)}\,/\,{\Gamma(x)}$.

Given that $x$ is an integer and the conditions $m_{10}, m_{11} \geq 0$, $m_{12} \geq 1$, and $m_{02} > 0$ hold, it follows that 
$$
\frac{\partial}{\partial x} \left\{\log \left(\frac{f_{X}(x | \bm)}{f_{Y}(x | \bm)}\right)\right\} \geq 0.
$$
This implies that 
%$Y$ is stochastically smaller than $X$ with respect to  likelihood ratio i.e., 
$X\geq_{\rm lr} Y$. For $X\geq_{\rm hr} Y$, $X\geq_{\rm mrl} Y$, and $X\geq_{\rm st} Y$, the proof is similar.
\end{proof} 

\begin{theorem}\label{thm:2}
 The BPGC distribution in  \eqref{eq:BPGC} is log-convex  if and only if $m_{11}>0$.
 \end{theorem}
 
\begin{proof} 
 Consider the negative of the  natural logarithm of the BPGC distribution as specified in \eqref{eq:BPGC}. Evaluating the second-order partial derivative  with respect to $x$ and $y$ yields:
$$ 
\frac{\partial}{\partial y}\left[\frac{\partial}{\partial x} \bigg(-\log\left( f(x,y|\bm)\right) \bigg)\right] = m_{11},
$$
where ${\partial}\,/\,{\partial x}$ and ${\partial}\,/\,{\partial y}$ denote the partial derivative operators with respect to $x$ and $y$, respectively.

For the distribution to be log-convex, this second-order partial derivative must be non-negative for all $(x,y)\in \mathbb{R^{+}}$. Therefore, log-convexity holds if $m_{11}>0$. 

Conversely, if $m_{11} \leq 0$, there are points where the second-order partial derivative may be negative, violating the condition of log-convexity.

Thus, the BPGC distribution is log-convex if and only if $m_{11}>0$.
\end{proof}

Next, the local dependence property of the BPGC distribution will be discussed. The local dependence function is defined as
$$
\xi \left( x, y \right) = \frac{{\rm \partial}^2}{\partial x \partial  y} \log\left(f(x,y|\bm)\right).
$$

\begin{theorem}\label{thm:3}
  The distribution $f\left(x,y|\bm\right)$  given by \eqref{eq:BPGC} is said to satisfy the reverse rule of order $2$ if and only if $m_{11} > m_{12}$.
\end{theorem}
\begin{proof}
To calculate the local dependence function, differentiate the logarithm of the function $f(x,y|\bm)$ with respect to $x$ and $y$. i.e.,
  \begin{equation}\label{eq:order2}
      \xi\left(x,y\right) = \frac{m_{12}}{y} - m_{11}.
  \end{equation}  
Next, if $m_{11}>m_{12},$ then the right-hand-side of \eqref{eq:order2} becomes non-positive for all $y\in \left(0, \infty\right)$. This implies that the distribution satisfies the reverse rule of order $2$. i.e., 
$$
\frac{m_{12}}{y} - m_{11} \le 0 \quad \forall y \in (0, \infty).
$$
Thus, this condition holds true if and only if $m_{11} > m_{12}$.
\end{proof}

The next theorem establishes the non-linearity of regression when $\left(X, Y\right)$ follows the BPGC distribution. Additionally, it characterizes this bivariate distribution.

\begin{theorem}\label{thm:4}
 If $(X,Y) \sim {\rm BPGC}(\bm)$, and if for all $x = 0, 1, \ldots$, 
$$
{\rm E}\left(Y | X = x\right) = \frac{m_{02} + m_{12}x }{m_{01} + m_{11} x y },
$$
and also if for all $y \in (0, \infty)$,
$$
{\rm E} \left(X | Y=y \right) = \exp\left( m_{10} - m_{11} y + m_{12} \log(y)\right);
$$
then, 
   \begin{itemize}
       \item the joint distribution of $\left(X, Y\right)$ is uniquely determined and follows the form given in \eqref{eq:BPGC}; 
       \item the regressions of $Y$ on $X$ and $X$ on $Y$ are both non-linear.
   \end{itemize}
\end{theorem}

\begin{proof}
With the parameters $\bm$ of the BPGC distribution, the joint distribution of $(X, Y)$ is defined by the conditional distributions in \eqref{eq:dist_X|Y} and \eqref{eq:dist_Y|X}. These conditions uniquely define the joint distribution of $(X, Y)$ as the BPGC distribution. 

To establish the non-linearity of the regression $Y$ on $X$, it is evident that the conditional expectation of $Y$ given $X$ is non-linear in $x$, which is expressed as the ratio ${(m_{02}+m_{12}x)}\,/\,{m_{01}+m_{11}xy}$. In a similar manner, the regression of $X$ on $Y$ reveals that the conditional expectation of $X$ given $Y$ is non-linear in $y$, characterized by the function $\exp\left(m_{10}-m_{11}y+m_{12}\log y\right)$. Hence, both regressions show non-linearity, clearly indicating that the relationship between $X$ and $Y$ inside the BPGC distribution is non-linear.
\end{proof}

\section{Parameter Estimation}\label{sec:Mle}
This section focuses on calculating the maximum likelihood estimators (MLEs) for the unknown parameters $\bm$ of the $\text{BPGC}(\bm)$ distribution, using a random sample of size $n$, denoted as $\left(\bX, \mathbf{Y}\right) = \{\left(X_1,Y_1\right), \left(X_2,Y_2\right), \ldots, \left(X_n,Y_n\right)\}$. Therefore, the log-likelihood function $ \ell \left(\bm |(\bX, \mathbf{Y}) \right)$ can be written as
\begin{eqnarray}\label{eq:MLE}
 \ell \left(\bm |(\bX, \mathbf{Y})\right) & = & \log \bigg( \prod_{i=1}^{n}X_{i} ! Y_{i} \bigg)^{-1}
 +n c + m_{10}\sum_{i=1}^{n}X_{i} -m_{11} \sum_{i=1}^{n}X_{i}Y_{i}\cr 
 && \cr 
 && \quad  - m_{01}\sum_{i=1}^{n}Y_{i}+ m_{02}\sum_{i=1}^{n}\log Y_{i} + m_{12}\sum_{i=1}^{n}X_{i}\log Y_{i} \, .
 \end{eqnarray}
The MLEs of the parameters are derived by taking partial derivatives of the log-likelihood function $\ell \left(\bm |(\bX, \mathbf{Y})\right)$ as shown as in \eqref{eq:MLE} with respect to $\bm$, and equating them to zero. However, the complexity of the underlying equations prevents analytically getting explicit forms of these estimators. Consequently, numerical techniques are required to determine the MLEs.

Constraints on the parameters are crucial in parameter estimation procedure for the BPGC distribution. The parameters $\bm = (m_{10}, m_{01}, m_{11}, m_{02}, m_{12})$ must satisfy the following conditions: $m_{01} > 0$, $m_{02} > 0$, $m_{10} > 0$, $m_{11} \ge 0$, and $m_{12} \geq 0$. Moreover, the normalizing constant $c$ is selected based on the other parameters. Due to these constraints, specialized optimization procedures are necessary for estimating the MLEs of the unknown parameters.

Barrier methods are highly effective for solving constrained optimization problems among the various optimization algorithms available. By incorporating a barrier component to the objective function, these methods penalize solutions that approach the boundaries of the feasible region. This penalty term ensures that the optimization process remains within the acceptable bounds, effectively establishing an invisible barrier that prevents constraint violations. 

The adaptive barrier algorithm, as outlined in \cite{Lange:2001}, is a sophisticated approach employed for the computation of MLEs. To find the MLEs of the parameters $m$, one can simply use the \texttt{mlEst} function available in this package.

Deriving closed-form expressions for the standard error of the MLEs presents theoretical challenges. This difficulty is particularly evident due to the complexity of the log-likelihood function in \eqref{eq:MLE}.
 
\section{Simulation Study}\label{sec:4}
This section provides details of the simulation study based on the methodology outlined in section \ref{sec:Mle}. The Gibbs sampling technique \citep{Gelfand:2000} is employed to generate random samples from the BPGC distribution, using the conditional distributions specified in \eqref{eq:dist_X|Y} and \eqref{eq:dist_Y|X}. The process is facilitated by the \texttt{rBPGC} function found in the \texttt{BPGC} library.

The MLEs derived using the \texttt{mlEst} function from the \texttt{BPGC} library for four different sets of model parameters are presented in Table \ref{tab:Sim1}, based on samples of varying sizes. As shown in Table \ref{tab:Sim1}, the estimates converge towards the true parameter values with an increase in sample size, as expected. More observations are necessary to achieve a more precise estimate when the parameter values are large; however, this also depends on the initial values that the generating algorithm takes. 

\begin{table}[!b]
\caption{MLEs of $\mathbf{m}$ for different sample sizes across four different cases from the BPGC distribution using Gibbs sampler method.}\label{tab:Sim1}
    \centering
    \resizebox{\textwidth}{!}{
    \begin{tabular}{@{}crccccc@{}}
    \hline
    $\bm = (m_{10}, m_{01}, m_{11}, m_{02}, m_{12})$ & $N$ & $\hat{m}_{10}$ & $\hat{m}_{01}$ & $\hat{m}_{11}$ & $\hat{m}_{02}$  & $\hat{m}_{12}$\\ 
    \hline 
    $(1,1,0.1,1,0.1)$ & 100 & 1.148076 & 0.546763 &  0.313887 & 0.820308 & 0.157127  \\
    & 1000 & 1.029098 & 0.881191 & 0.114075 & 0.945462 & 0.096011\\
    & 10000 & 1.007424 & 0.949148 & 0.105833 & 0.949690 & 0.105854\\
    & 100000 & 0.996931 & 1.008777 & 0.097689 & 1.000099 & 0.099262 \\
    \hline 
     $(1,1,1,1,1)$  & 100 & 1.607437 & 0.695712 &  1.728029 & 0.847946 & 1.317043  \\
     & 1000 & 1.250769 & 0.851429 & 1.181627 & 0.892492 & 1.159671\\
     & 10000 & 1.026660 & 0.950852 & 1.016816 & 0.952526 & 1.029666\\
     & 100000 & 1.004788 & 1.000320 & 1.002211 & 0.997736 & 1.009262 \\
     \hline 
     $(1,5,1,5,1)$ & 100 & 3.139724 & 2.977665 & 3.070716 & 3.290903 & 2.653367  \\
     & 1000 & 1.424880 & 4.455080 & 1.377309 & 4.588816 & 1.341538\\
     & 10000 & 1.039891 & 4.779689 & 1.040329 & 4.759322 & 1.073789\\
     & 100000 & 0.990017 & 4.973420 & 0.99233 & 4.980975 & 0.984525\\ 
     & 500000 & 1.006570 & 4.996552 & 1.004323 & 4.995206 & 1.003916 \\
     \hline 
     $(5,5,5,5,5)$ & 100 & 7.183235 & 3.270673 & 7.190890 & 3.490521 & 6.927195 \\
     & 1000 & 5.326312 & 4.580192 & 5.284450 & 4.697007 & 5.202252\\
     & 10000 & 4.998869 & 4.759772 & 4.998852 & 4.756931 & 4.997521\\
     & 100000 & 4.959363 & 4.978258 & 4.961495 & 4.983204 & 4.949159 \\ 
     & 500000 & 5.031970 & 4.999041 & 5.028674 & 4.996252 & 5.027729\\
     \hline 
    \end{tabular} 
    }
\end{table}

The 3D-histogram of simulated data is displayed in Figure \ref{fig:Sim1}, alongside the fitted and true BPGC distributions for the cases in Table \ref{tab:Sim1}, using $N = 1000$ observations. The BPGC distribution estimated by the MLEs (blue), closely matches the BPGC distribution fitted with the true parameters (red). As shown in Figure \ref{fig:Sim1}, a good fit between the estimated BPGC distribution and the empirical distribution is evident. 

To validate the accuracy of our fitted BPGC distribution against empirical data, a multivariate two-sample Kolmogorov-Smirnov (KS) test, as described in \cite{Fasano:Franceschini:1987}-- commonly referred to as the two-sample Fasano-Franceschini (FF) test -- is employed.  This test evaluates whether the BPGC distribution, fitted via maximum likelihood estimation, represents the underlying distribution of observed data. The procedure involves:
\begin{enumerate}
    \item[(1)] The original dataset is used as the first sample, representing the empirical distribution (denoted as $F_1$) of observed data.
    \item[(2)] The parameters of the BPGC distribution are estimated using the MLE method, yielding the estimate $\hat{\bm}$.
    \item[(3)] A second sample is simulated from the BPGC distribution with the estimated parameters $\hat{\bm}$, i.e., $F_2 = \text{BPGC}(\hat{\bm})$.
    \item[(4)] The two-sample Fasano-Franceschini (FF) test is applied to evaluate the null hypothesis $\mathcal{H}_{0}: F_1 = F_2$ against the alternative $\mathcal{H}_1: F_1 \neq F_2$, assessing whether the distributions of $F_1$ and $F_2$ significantly differ.
\end{enumerate}

The \texttt{FFtest} function within the \texttt{BPGC} library automates these steps by handeling parameter estimation, simulating the second sample, and performing the FF test, thus providing a comprehensive evaluation of how goodness-of-fit of the BPGC model to the observed data.

The test statistics and $p$-values are presented in Table \ref{tab:sim2}, showing that the fitted distribution using the MLEs is closely aligned with the distribution using the true parameter values.

\begin{table}[!ht]
    \centering
    \caption{The test-statistics and $p$-values}
    \label{tab:sim2}
    \begin{tabular}{@{}lcc@{}}
    \hline 
    $\bm = (m_{10}, m_{01}, m_{11}, m_{02}, m_{12})$ & test statistic & $p$-value \\ 
    \hline 
    $(1,1,0.1,1,0.1)$ & 568000 & 0.8637 \\ 
    $(1,1,1,1,1)$ & 616000 & 0.6975 \\
    $(1,5,1,5,1)$ & 531000 & 0.9087  \\
    $(5,5,5,5,5)$ & 549000 & 0.8644 \\ 
    \hline 
    \end{tabular}    
\end{table}

\begin{figure}[!ht]
    \centering
    \begin{tabular}{cc}
        \includegraphics[scale = 0.5, trim={40mm 0mm 40mm 0mm},clip]{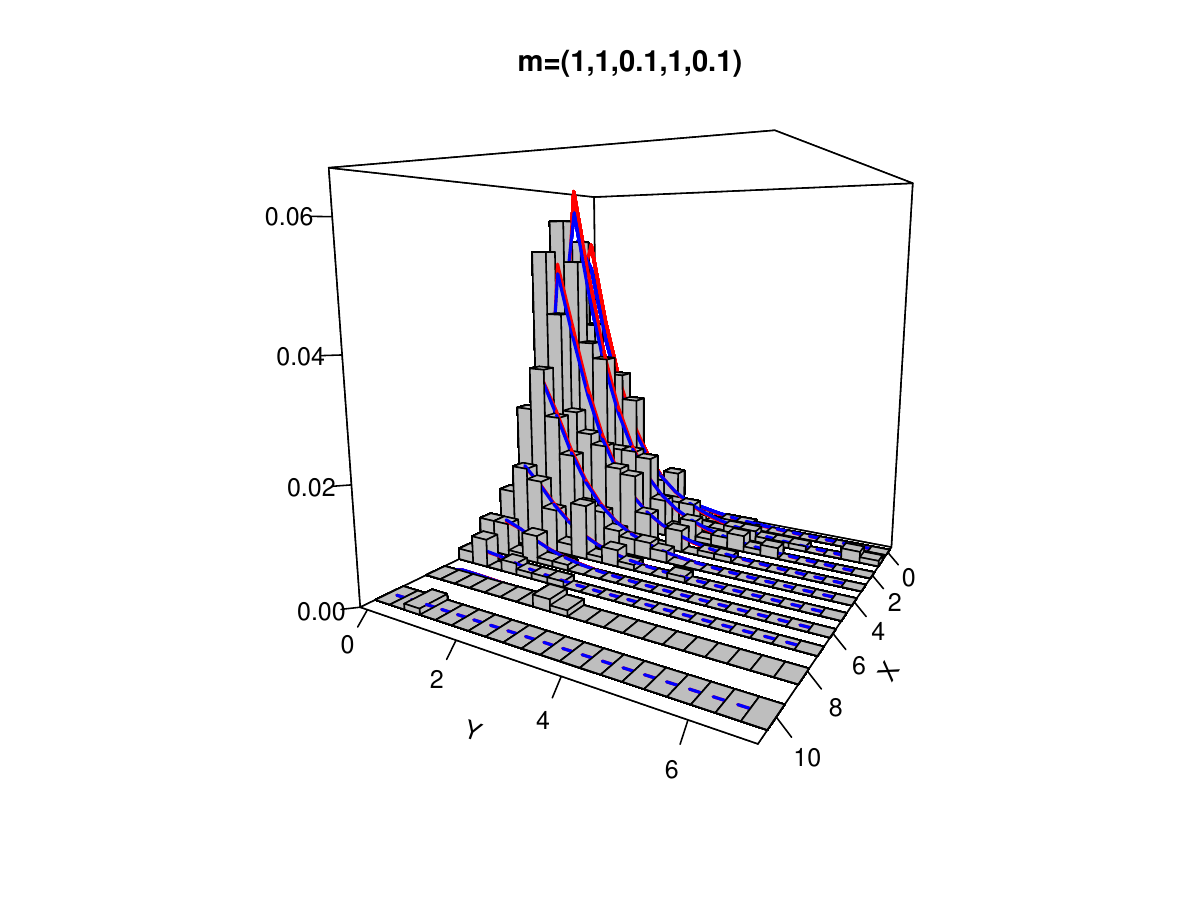} &       
        \includegraphics[scale = 0.5, trim={40mm 0mm 40mm 0mm},clip]{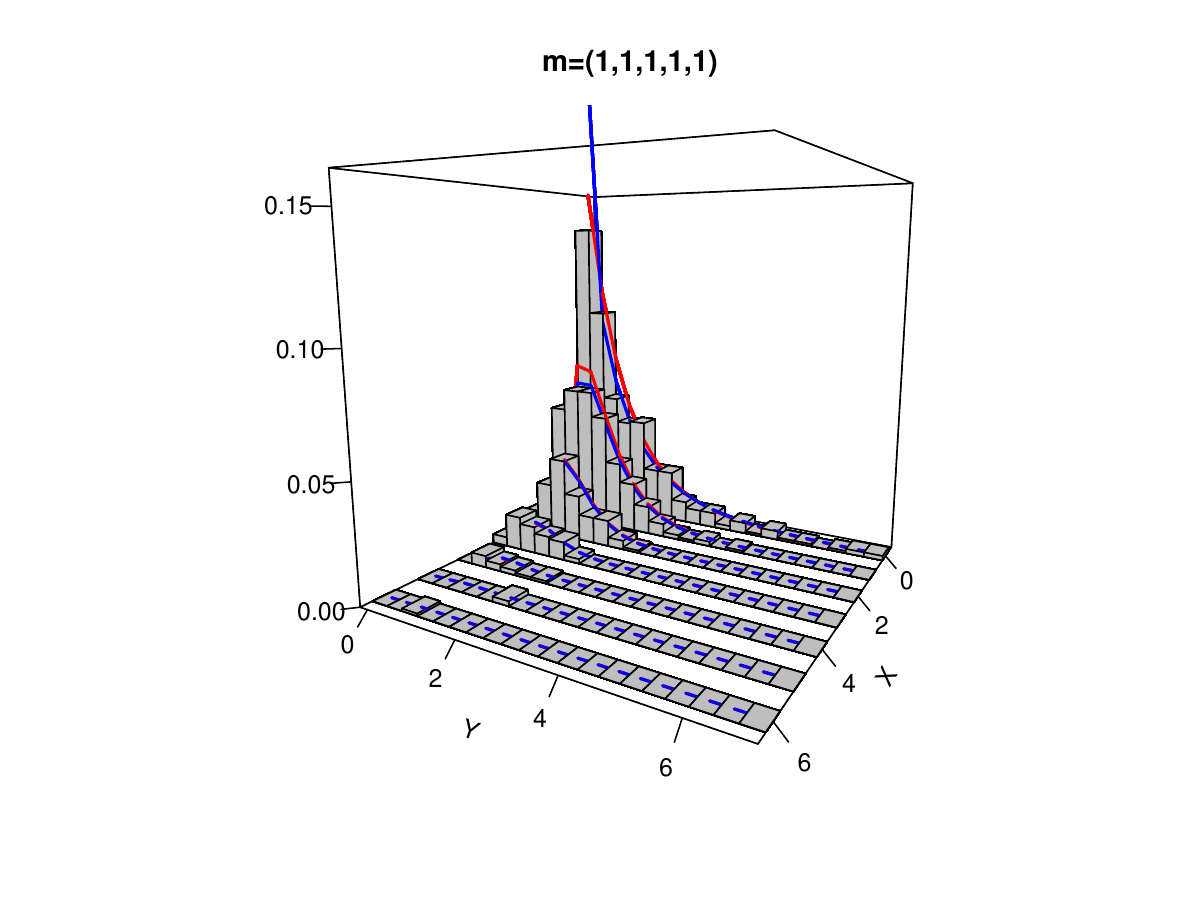} 
        \\        
        \includegraphics[scale=0.5, trim={40mm 0mm 40mm 0mm},clip]{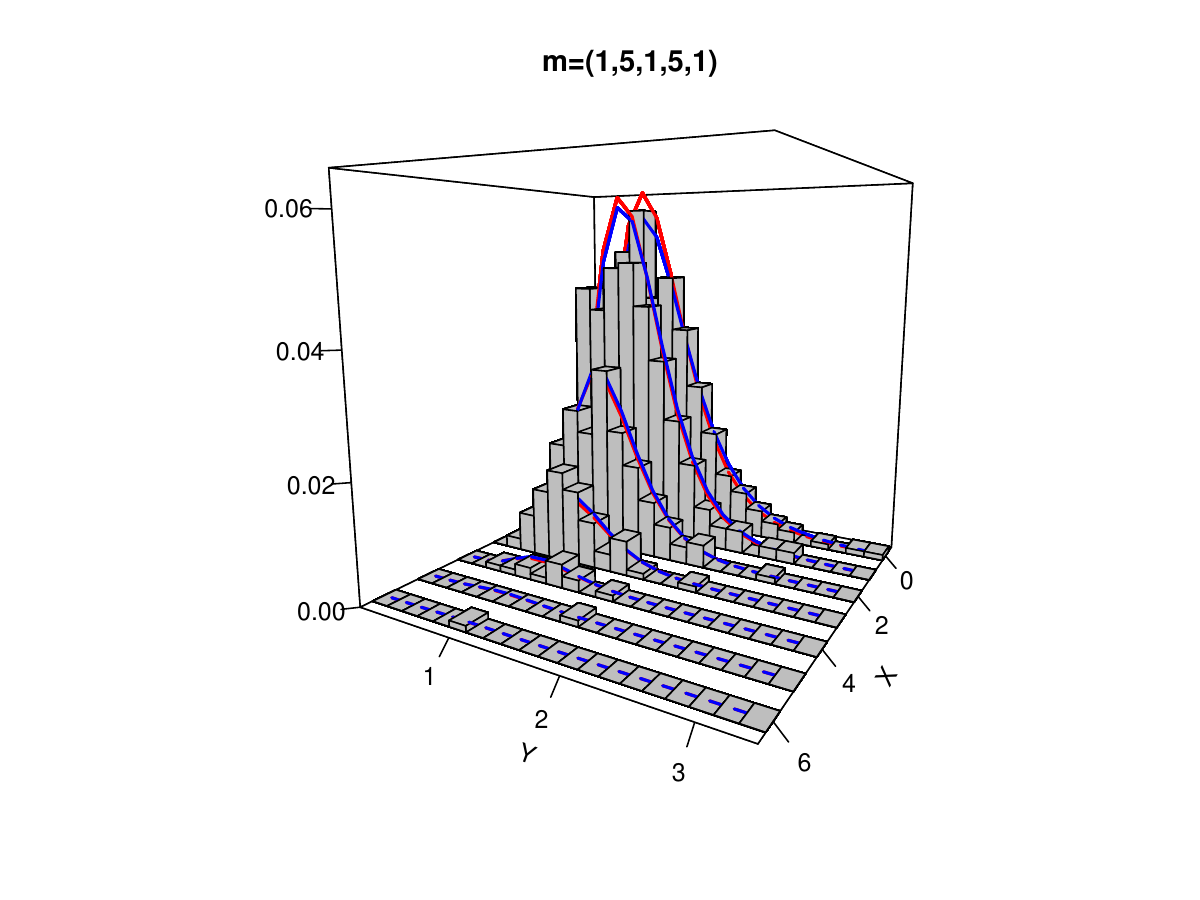} & 
        \includegraphics[scale=0.5, trim={40mm 0mm 40mm 0mm},clip]{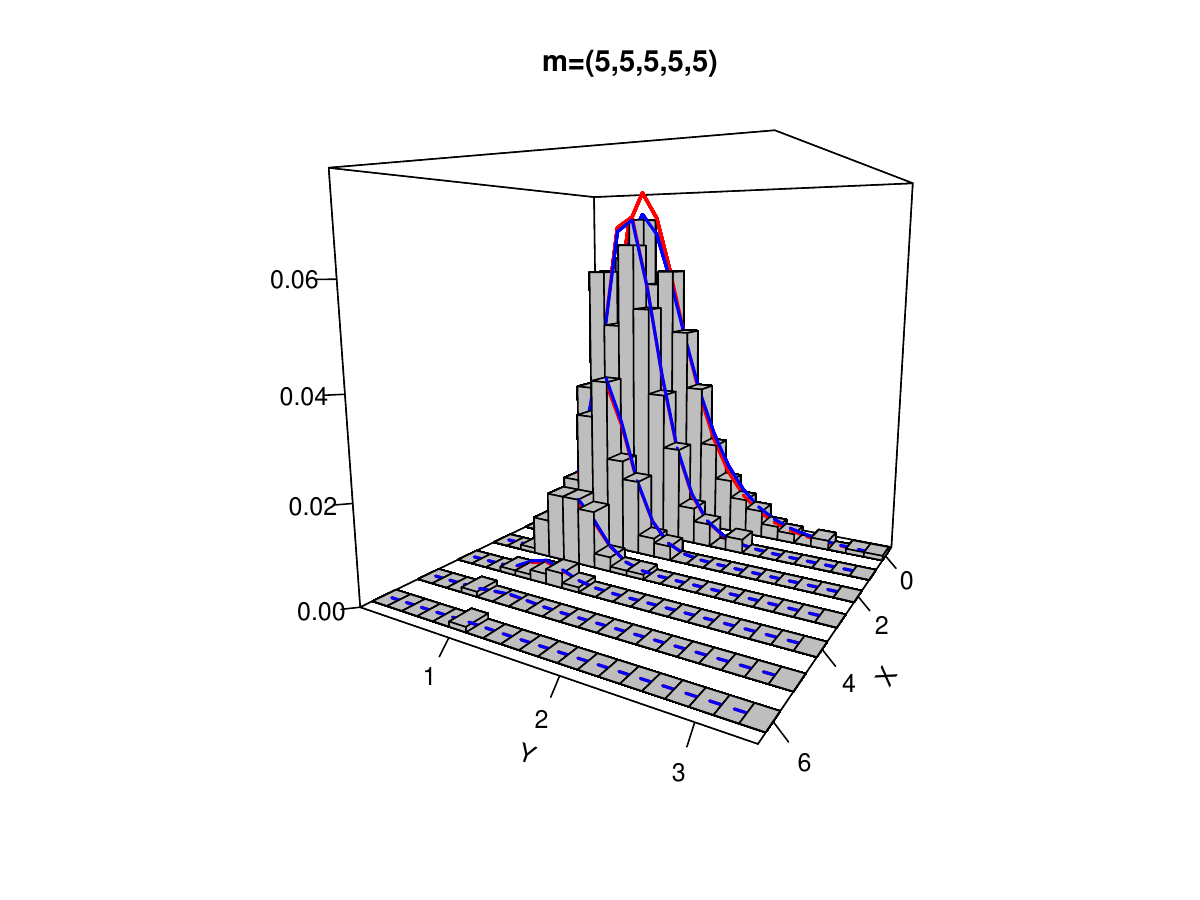}\\
    \end{tabular}
     \caption{The 3D histogram fo the simulated, the true (red) and the fitted (blue) BPGC distribution.}\label{fig:Sim1}
    \label{fig:1}
\end{figure}

\section{Application}\label{sec:5}
In this section, we apply the BPGC model to a toy dataset concerning hospital admissions and treatment costs, available as \texttt{hospitalData} in the \texttt{BPGC} package. The dataset comprises $500$ observations, where each observation represents a single day. The variable `number\_of\_admissions' (denoted by $X$) corresponds to the number of admissions per day, while `total\_cost\_of\_treatment' (denoted by $Y$) represents the total treatment costs incurred per day. Table \ref{tab:app1} provides descriptive statistics for this dataset. Figure \ref{fig:app1} presents a bar plot depicting the distribution of the number of admissions ($X$) and a histogram showing the distribution of the total treatment costs ($Y$). In addition, Figure \ref{fig:app2} illustrates the histogram overlaid with a fitted probability density curve to visualize the goodness-of-fit of the BPGC model to the observed data.

\begin{table}[!ht]
    \centering
    \caption{Descriptive measures of the hospital dataset}\label{tab:app1}
    \begin{tabular}{@{}ccccccc@{}}
    \hline 
    & Min & $Q_{0.25}$ & Median & Mean & $Q_{0.75}$ & Max \\
    \hline 
    $X$ & $2.000$ & $8.000$ & $10.000$ & $9.856$ & $12.000$ & $20.000$\\ 
    $Y$ & $0.821$ & $8.250$ & $13.614$ & $14.578$ & $19.181$ & $53.671$\\
    \hline 
    \end{tabular}
\end{table}

\begin{figure}[!ht]
    \centering
    \includegraphics[height = 10cm, width = \textwidth]{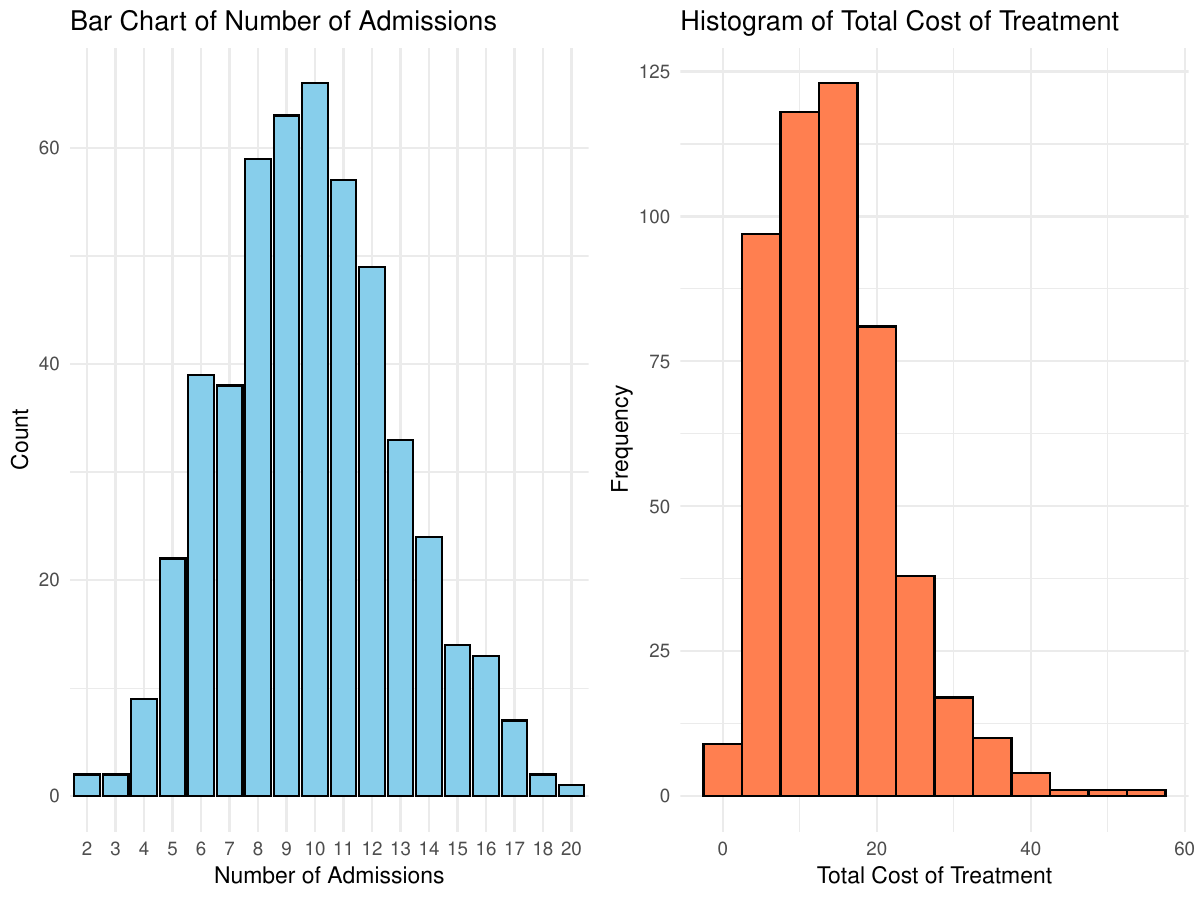} 
     \caption{The Barplot of number of admissions (left) and the histogram of total cost of treatment (right).}\label{fig:app1}
\end{figure}

Using the \texttt{mlEst} function from the \texttt{BPGC} library, the estimated values for $\mathbf{m}$ are obtained as
\begin{equation}\label{eq:est_m}
    \hat{m}_{10} = 2.1809, \quad \hat{m}_{01} = 0.1880 \quad \hat{m}_{11} = 0.0018 \quad \hat{m}_{02} = 2.4806, \quad \hat{m}_{12} = 0.0535
\end{equation}

To evaluate the adequacy of the BPGC distribution's fit, the two-sample Fansano-Franceschini test described in the preceding section is applied using the \texttt{FFtest} function. This test compares the empirical distribution of the observed data with a simulated sample from the estimated BPGC distribution. The results of the test are as follows:
\begin{equation}\label{eq:results}
    \text{test statistic} = 381500 \quad \text{and} \quad \text{$p$-value} = 0.9307.
\end{equation}
With a $p$-value of $0.9307$, the fitted BPGC distribution does not differ significantly from the observed data. The high $p$-value indicates that the hospital admissions and treatment cost data are well-fitted by the BPGC model.

\begin{figure}[!ht]
    \centering
    \includegraphics[scale = 0.875, trim={30mm 10mm 40mm 10mm},clip]{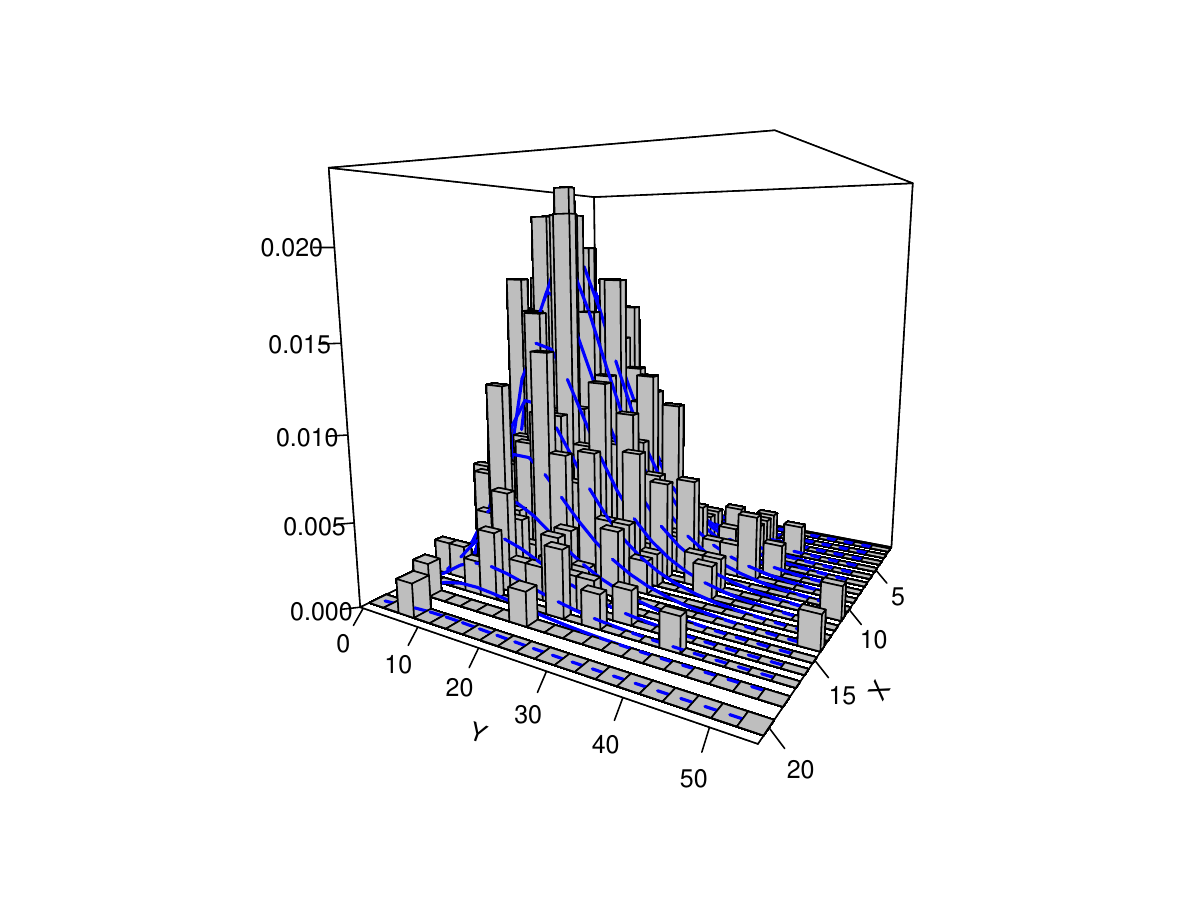} 
     \caption{The 3D histogram for the hospital dataset with the fitted (blue) BPGC distribution.}\label{fig:app2}
\end{figure}

Based on the results presented in \eqref{eq:results}, where the test statistic is $381500$ and the corresponding $p$-value is $0.9307$, and considering the fitted curve illustrated in Figure \ref{fig:app2}, which was generated using the estimated parameters in \eqref{eq:est_m}, it can be confidently asserted that the BBCD model provides an excellent fit to the hospital data.

\section{Concluding Remarks}\label{sec:6}
There are several well-established bivariate mixtures distributions in the rich literature of distribution theory, either both continuous or discrete. Interestingly, relatively few distributions arise from a mixture of a discrete and a continuous variable. In this paper, we provide a detailed discussion of a joint distribution of such mixtures, namely the BPGCdistribution. In this distribution, among the two conditionals, one is discrete (Poisson), while the other is absolutely continuous (Gamma), with parameters that depend on the conditioning variable(s), respectively.

The foundation of this distribution can be found in \citet{Arnold:1999}, yet previous discussions on the structural properties of the BPGC distribution have been relatively limited until this study. Since the BPGC distribution belongs to the expansive class of exponential family of distributions, it enjoys several useful structural properties, some of which have been explored in this paper.  Despite its computational complexity, it is possible to consider a multivariate extension. 

A major question that might arise concerns the possible applicability of such a model in a practical setting. We are currently working on this issue, the results will be published in a future study.

\section*{Acknowledgments}
The work of Filipe J. Marques and Mina Norouzirad is funded by national funds through the FCT – Funda\c{c}\H{a}o para a Ci\^{e}ncia e a Tecnologia, I.P., under the scope of the projects UIDB/00297/2020 (\url{https://doi.org/10.54499/UIDB/00297/2020}) and UIDP/00297/2020 (\url{https://doi.org/10.54499/UIDP/00297/2020}) (Center for Mathematics and Applications)”.

\end{document}